\documentclass[onecolumn,12pt,draft]{ieeetran}
\usepackage{amsmath}
\usepackage[final]{graphicx}
\usepackage{amsmath,epsfig,cite}
\usepackage{graphicx}
\usepackage{amssymb}
\usepackage{mathrsfs}

\newtheorem{thry}{Theorem}
\newtheorem{prop}{Proposition}

\newtheorem{defn}{Definition}
\usepackage[dvips,dvipdfm]{geometry}
\usepackage{setspace}
\title{Applications of Tauberian Theorem for High-SNR Analysis of Performance over Fading Channels}
\author{Yuan Zhang, Cihan Tepedelenlio\u{g}lu, \emph{Member, IEEE}
\thanks{The authors are with the School of Electrical, Computer,
and Energy Engineering, Arizona State University, Tempe, AZ 85287,
USA. (Email: yzhang93@asu.edu, cihan@asu.edu).} }
\date{}

\begin{document}

\maketitle

\vspace{-0.65in}

\begin{abstract}
\begin{spacing}{1.5}
This paper derives high-SNR asymptotic average error rates over
fading channels by relating them to the outage probability, under
mild assumptions. The analysis is based on the Tauberian theorem for
Laplace-Stieltjes transforms which is grounded on the notion of
regular variation, and applies to a wider range of channel
distributions than existing approaches. The theory of regular
variation is argued to be the proper mathematical framework for
finding sufficient and necessary conditions for outage events to
dominate high-SNR error rate performance. It is proved that the
diversity order being $d$ and the cumulative distribution function
(CDF) of the channel power gain having variation exponent $d$ at $0$
imply each other, provided that the instantaneous error rate is
upper-bounded by an exponential function of the instantaneous SNR.
High-SNR asymptotic average error rates are derived for specific
instantaneous error rates. Compared to existing approaches in the
literature, the asymptotic expressions are related to the channel
distribution in a much simpler manner herein, and related with
outage more intuitively. The high-SNR asymptotic error rate is also
characterized under diversity combining schemes with the channel
power gain of each branch having a regularly varying CDF. Numerical
results are shown to corroborate our theoretical analysis.
\end{spacing}
\end{abstract}

\vspace{-0.2in}


\section{Introduction} \label{intro}

The concept of diversity plays a critical role in communications
over fading channels by quantifying the decrease of the average
error rate as a function of the average SNR \cite{goldsmithbook05}.
In \cite{wang03}, the authors show analytically that if the channel
power gain has a probability density function (PDF) which behaves
asymptotically like a multiple of its $(d-1)$-th power near $0$, the
diversity order is $d$. The approach employed in \cite{wang03} is
also utilized in \cite{ordoez07} to analyze the SNR shift in more
complicated systems with the same diversity order, and in
\cite{zhao07} to derive the optimal diversity and multiplexing
trade-off for generalized fading channels. A similar definition for
diversity order is also established in \cite{zheng03}, and used to
determine the diversity order associated with some lattice-based
MIMO detection schemes \cite{taherzadeh07a,taherzadeh10} with equal
numbers of transmit and receive antennas.

\vskip 1.0mm

Regular variation is a concept in real analysis which describes
functions exhibiting power law behavior asymptotically near zero, or
infinity, and is applied to several different areas including
probability theory \cite{fellerbook71,binghambook89}. The Tauberian
theorem for Laplace-Stieltjes transforms \cite[p.37]{binghambook89}
asserts that if a function with the non-negative support is
regularly varying at the origin (infinity), then then its
Laplace-Stieltjes transform must be regularly varying at infinity
(origin). The applications of this in communications and networking
are primarily seen in asymptotic queueing analysis
\cite{zwart01,jelenkovic01}. In addition, fading channel
distributions with properties related to regular variation have been
studied with emphasis on scaling properties of ergodic channel
capacity, under several communication scenarios involving channels
with heavy tail behavior in \cite{gesbert11}, where capacity scaling
of systems rather than diversity analysis is considered.

\vskip 1.0mm

To the best of our knowledge, diversity-related performance analysis
over fading channels has never been addressed using the Tauberian
theorem together with the theory of regular variation. In this
paper, we perform analysis of diversity, and more generally the
high-SNR asymptotic error rate for a wide range of channel
distributions and modulation types under mild assumptions on the
instantaneous error rate and channel distribution functions. We
prove that the diversity order is $d$ if and only if the CDF of the
channel power gain has variation exponent $d$ at the origin,
provided that the instantaneous error rate is upper bounded by an
exponential function of the instantaneous SNR. Furthermore, we
derive more explicit closed-form expressions of the asymptotic error
rates for the special cases of practical instantaneous error rates
that capture systems including, but not limited to, $M$-PSK and
square $M$-QAM. We also establish a unified approach to characterize
asymptotic error rate performance for diversity combining schemes,
with the channel power gain of each diversity branch having a
regularly varying CDF. The results in this paper establish a
mathematical framework for determining the conditions under which
the outage event dominates the error rate performance. Compared to
existing approaches, our asymptotic average error rate
characterization applies to a more general set of channel
distributions, is related to the channel CDF in a simple manner, is
intuitively linked with outage, and draws from the well-established
mathematical theory of regular variation. Compared to the conference
version \cite{ow11b}, this paper contains complete proofs, and more
general assumptions on the instantaneous error rate. Furthermore,
expanded numerical results not included in \cite{ow11b} illustrate
improved accuracy compared to \cite{wang03}.

In Section \ref{math_pre}, the mathematical preliminaries regarding
regular variation and Tauberian theorem are presented. In Section
\ref{chnl_sys_mdl} we establish the channel and system model through
the assumptions on the instantaneous error rate function and the
channel distribution. Section \ref{div_defs} establishes a
definition of diversity order based on regular variation, which,
under general conditions, is proved to be equivalent to the most
general definition used in the literature. In Section \ref{asympt}
we go beyond diversity analysis and establish the asymptotic
equivalence between the average error rate and outage. The
asymptotic average error rate expressions for diversity combining
schemes are derived in Section \ref{div_comb}, and Section
\ref{concl} concludes the paper.

\section{Mathematical Preliminaries} \label{math_pre}

Before the mathematical preliminaries, we have a few remarks about
notations. Asymptotic equality $H_1(x) \sim H_2(x)$ as $x
\rightarrow a$ means that $\lim_{x \rightarrow a} H_1(x)/H_2(x)=1$,
and $H_1(x)= o(H_2(x))$ as $x \rightarrow a$ means that $\lim_{x
\rightarrow a} H_1(x)/H_2(x)=0$. ${\rm E}[\cdot]$ denotes the
expectation over the channel distribution, $\log (\cdot)$ denotes
natural logarithm, and $:=$ denote equivalence by definition.
Finally, $\gamma(a,b):= \int_0^b t^{a-1} e^{-t} dt$, $\Gamma(a):=
\gamma(a,\infty)$, $Q(x)= (2 \pi)^{-1/2} \int_x^{\infty} \exp
(-u^2/2) du$, and $K_{\nu}(\cdot)$ denotes modified Bessel function
of the second kind of order $\nu$.

We next provide a brief sketch of the mathematical concepts
(slow/regular/rapid variation) and theorems (a representation
theorem and the Tauberian theorem), as well as related propositions,
which are drawn from standard mathematical references such as
\cite{fellerbook71,binghambook89}. \vskip 1mm

\begin{defn}
A real valued function $H(x)$: $\mathbb{R}^+ \rightarrow
\mathbb{R}^+$ is called a {\it Karamata function} (at $\infty$ or
$0$), if $\lim_{x \rightarrow \infty} H(tx)/H(x)$ or $\lim_{x
\rightarrow 0} H(tx)/H(x)$ exists for $t>0$. Specifically, the limit
must be in the form of $t^m$, where $m \in [-\infty,\infty]$ is
called the {\it variation exponent} of $H(x)$. \vskip 1mm
\end{defn}

Throughout the paper we interpret $t^{-\infty}=0$ and
$t^{\infty}=\infty$ for $t>1$, and vice versa when $t<1$.
Specifically, $H(x)$ is termed slowly/regularly/rapidly varying if
$m=0$, $0<|m|< \infty$, and $|m|=\infty$, respectively. By
definition, $G(x)=H(x^{-1})$ is slowly/regularly/rapidly varying at
$0$ if and only if $H(x)$ is slowly/regularly/rapidly varying at
$\infty$. For example, with $0<m< \infty$, $x^m$ is regularly
varying with exponent $m$ at both $0$ and $\infty$; $x^m e^{-x}$ is
regularly varying with exponent $m$ at $0$ and rapidly varying with
exponent $-\infty$ at $\infty$; $\log (1+x^m)$ is regularly varying
with exponent $m$ at $0$ and slowly varying at $\infty$.

For $H(x)$ being slowly varying at $\infty$, there is a
representation theorem \cite[Theorem 1.3.1, p.12]{binghambook89}
with the proof available therein. More generally, this theorem can
be extended to handle the cases of regular/rapid variation, as
described briefly in \cite[p.21]{binghambook89} and summarized in
\cite{dembinska06}. We express it in our notations as the following,
which will be useful in establishing the relations among the
different definitions of diversity order in Section \ref{div_defs}.
\vskip 1mm
\begin{thry} \label{rep_thm}
$H(x)$ has variation exponent $m$ at $\infty$, i.e. $\lim_{x
\rightarrow \infty} H(tx)/H(x)=t^m$ for $m \in [-\infty,\infty]$, if
and only if $H(x)= c(x) \exp \left\{ \int_a^x (\epsilon(u)/u) du
\right\}$ for some $a>0$, $c(x) \rightarrow c \in (0,\infty)$, and
$\epsilon(x) \rightarrow m$, as $x \rightarrow \infty$. \vskip 1mm
\end{thry}

It is well-known that the MGF of the channel random variable is
central in average performance over fading channels
\cite{simonbook00}. Instead of the MGF, the Laplace-Stieltjes
transform can also be considered, and the theory behind this
transform can be brought to bear. The primary tool used in this
paper is the Tauberian theorem for Laplace-Stieltjes transforms,
which relates the asymptotic properties of a function and of its
Laplace-Stieltjes transform, with both properties characterized by
slow/regular/rapid variation. This function will often (but not
always) be the CDF of the channel power gain in the sequel. We
summarize \cite[Theorem 1, p.443]{fellerbook71} and \cite[Theorem 2,
p.445]{fellerbook71} into the following, with the proofs available
therein.
\begin{thry} \label{tbr_th}
If a function $H(x) \geq 0$ defined on $x \geq 0$ has a
Laplace-Stieltjes transform $\mathcal{L}(s)= \int_0^{\infty} e^{-sx}
dH(x)$ for $s \geq 0$, then for $|m| \in [0,\infty]$, $H(x)$ having
variation exponent $m$ at $\infty$ (or $0$) and $\mathcal{L}(s)$
having variation exponent $-m$ at $0$ (or $\infty$) imply each
other. In addition, for $|m|<\infty$ and $l(x)$ being slowly varying
at $0$ (or $\infty$), the relations $H(x) \sim x^m l(x)$ as $x
\rightarrow 0$ (or $x \rightarrow \infty$) and $\mathcal{L}(s) \sim
\Gamma(m+1) s^{-m} l(s^{-1})$ as $s \rightarrow \infty$ (or $s
\rightarrow 0$) imply each other.
\end{thry}

We also have the following proposition which follows from
\cite[p.27]{binghambook89}, and will be used to establish the
equivalence of the CDF- and PDF-based assumptions on the channel.
\begin{prop} \label{prop_cdf_pdf}
If $H(x)$ is differentiable with $h(x)=d H(x)/dx$, and $h(x)= m
x^{m-1} l(x)$ for some function $l(x)$ slowly varying at $0$, then
$H(x) \sim x^m l(x)$ as $x \rightarrow 0$.
\end{prop}
The next proposition will be useful to link the average error rate
and outage.
\begin{prop} \label{prop_sl1}
If $H(x)$ is slowly varying at $\infty$ (or $0$), and bounded
between two constants $h_1$ and $h_2$ (where $0 \leq
h_1<h_2<\infty$) for sufficiently large (small) $x$, then it
converges to a constant in $[h_1,h_2]$ as $x \rightarrow \infty$ (or
$x \rightarrow 0$).
\end{prop}
\begin{proof}
See \ref{prop_sl1_proof}.
\end{proof}

\section{Channel and System Model} \label{chnl_sys_mdl}

We consider average performance for systems over fading channels
where the channel can be captured by an instantaneous SNR random
variable. These include SISO systems, SIMO systems with diversity
combining, and MISO systems with beamforming. To facilitate
subsequent derivations, we factor the received instantaneous SNR
into channel-independent and channel-dependent components by
expressing it as a product $\rho z$, where the average SNR (per
symbol) $\rho$ is deterministic, and $z$ is the channel power gain
random variable having CDF $F(z)$, resulting from both the channel
and the system setup (e.g. beamforming or diversity combining). The
instantaneous error rate ${\rm P_e} (\rho z)$, as a function of
$\rho z$, is determined by the modulation type together with the
noise distribution and represents bit, or symbol error rate.

In this paper, we focus on analyzing the average error rate
\begin{equation}
{\rm \overline{P}_e}(\rho):= {\rm E} [{\rm P_e} (\rho z)]=
\int_0^{\infty} {\rm P_e} (\rho z) dF(z)
\end{equation}
for large $\rho$. We make one of the following {\it assumptions on
the instantaneous error rate} in our subsequent derivations 
\begin{itemize}
\begin{item}
{\bf AS1a} ${\rm P_e} (\rho z) \leq \beta e^{-\alpha \rho z}$
$\forall$ $\rho$, where $\alpha, \beta \in (0,\infty)$ are
constants. In other words, ${\rm P_e} (\rho z)$ is upper-bounded by
an exponential function of $\rho z$;
\end{item}
\begin{item}
{\bf AS1b} ${\rm P_e} (\rho z)= \beta e^{-\alpha \rho z}$, where
$\alpha, \beta \in (0,\infty)$ are constants;
\end{item}
\begin{item}
{\bf AS1c} ${\rm P_e} (\rho z)= \int_0^{\theta_{\rm t}} g_2(\theta)
\exp(-\rho z/g_1(\theta)) d \theta$, where $\theta_{\rm t} \in
(0,\pi)$, and the corresponding $g_1(\theta)$ and $g_2(\theta)$ are
both finite non-negative. In other words, ${\rm P_e} (\rho z)$ is a
positive mixture of decreasing exponential functions of $\rho z$.
\end{item}
\end{itemize}
It can be easily shown that {\bf AS1b} $\Rightarrow$ {\bf AS1c}
$\Rightarrow$ {\bf AS1a}. We will have results of differing
generality corresponding to these assumptions. {\bf AS1a} holds for
all modulation types with the additive noise being Gaussian. {\bf
AS1b} is a special case of {\bf AS1a} typically applying to
non-coherent modulations, e.g. ${\rm P_e} (\rho z)= (1/2) e^{-\rho
z}$ is the BER for DPSK. {\bf AS1c} applies to many practical
settings with appropriate choices of $\theta_{\rm t}$, $g_1(\theta)$
and $g_2(\theta)$. Specifically, the SER of $M$-PSK
\cite[(8.22)]{simonbook00} fits {\bf AS1c} with $\theta_{\rm t}= (1-
M^{-1}) \pi$, $g_1(\theta)= \sin^2 \theta/ \sin^2 (M^{-1} \pi)$ and
$g_2(\theta)= \pi^{-1}$; by making use of $Q^{\alpha}(x)= \pi^{-1}
\int_0^{\pi/2 \alpha} \exp (-x^2/(2 \sin^2 \theta)) d \theta$ for
$\alpha=1,2$ \cite[(4.2) and (4.9)]{simonbook00}, it is possible to
express the BER ${\rm P_e} (\rho z)= Q(\sqrt{2 \rho z})$ for BPSK
and the SER ${\rm P_e} (\rho z)= 4 (1-M^{-1/2}) Q(\sqrt{3 \rho
z/(M-1)})- 4 (1-M^{-1/2})^2 Q^2(\sqrt{3 \rho z/(M-1)})$
\cite[(8.10)]{simonbook00} for square $M$-QAM in the form of {\bf
AS1c}. For an arbitrary two-dimensional signal constellation with
polygon-shaped decision regions over AWGN, ${\rm P_e} (\rho z)$ can
be treated as linear combination of the error rates associated with
individual constellation points given by \cite[(5.71)]{simonbook00},
which is also in the form of {\bf AS1c}.

We make one of the following {\it assumptions on the channel} in
terms of the CDF $F(z)$, which is assumed not to depend on the
average SNR $\rho$:
\begin{itemize}
\begin{item}
{\bf AS2a} $\lim_{z \rightarrow 0} F(\tau z)/F(z)= \tau^d$, where
$\tau, d>0$. In other words, $F(z)$ is a Karamata function of $z$ at
$0$ with variation exponent $d \in (0,\infty]$;
\end{item}
\begin{item}
{\bf AS2b} Same as {\bf AS2a} with $d \in (0,\infty)$, in which case
one can write $F(z)=z^d l(z)$ with $l(z)$ slowly varying at $0$;
\end{item}
\begin{item}
{\bf AS2c} The PDF $f(z):= dF(z)/dz$ exists and $f(z)= d z^{d-1}
l_1(z)$, where $d \in (0,\infty)$ and $l_1(z)$ is slowly varying at
$0$.
\end{item}
\end{itemize}
Based on Proposition \ref{prop_cdf_pdf}, {\bf AS2b} is implied by
{\bf AS2c}. Therefore the assumptions, as listed, get stronger,
i.e., {\bf AS2c} $\Rightarrow$ {\bf AS2b} $\Rightarrow$ {\bf AS2a}.
The difference between {\bf AS2a} and {\bf AS2b} is the allowance of
$d=\infty$, which is equivalent to the rapid variation of $F(z)$ at
the origin. This holds, for example, for log-normal shadowing. When
$d=\infty$ is ruled out, we have {\bf AS2b}, which will be seen to
offer sharper results than offered by {\bf AS2a}. When the PDF of
the channel power gain exists and is in a simpler form than the CDF,
{\bf AS2c} can be employed. This assumption is similar to (but more
general than) that of \cite{wang03}, which is also based on the PDF.


As mentioned, it can be verified that log-normal shadowing for which
$\log z$ is Gaussian has a CDF $F(z)$ satisfying {\bf AS2a} with
$d=\infty$. Also, it is easy to see that Rayleigh, Nakagami-$K$
(Ricean) and Nakagami-$q$ (Hoyt) fading channels satisfy all three
assumptions with $d=1$. Furthermore, Nakagami-$m$ fading with
$f(z)=(\Gamma(m))^{-1} m^m z^{m-1} \exp(-mz)$ and $F(z)= \gamma
(m,mz)$ can be verified to satisfy all three assumptions with $d=m$;
Weibull fading defined by $f(z)=k z^{k-1} \exp(-z^k)$ or by $F(z)=
1-\exp(-z^k)$ has $d=k$. Consider also generalized-$K$ fading which
models composite multi-path fading and shadowing \cite{bithas06}.
For such fading with $f(z)= 2 (km)^{(k+m)/2} z^{(k+m-2)/2} K_{k-m}
[2 (kmz)^{1/2}]/(\Gamma(m) \Gamma(k))$ \cite[eqn. (2)]{bithas06},
$m,k>0$, based on $K_0 (x) \sim -\log (x/2)- \gamma_{\rm em}$
($\gamma_{\rm em} \approx 0.5772$ is the Euler-Mascheroni constant)
and $K_{\nu} (x) \sim \Gamma(|\nu|) (2 x^{-1})^{|\nu|} /2$ for $\nu
\ne 0$ near $x=0$, it can be verified that {\bf AS2c} holds with $d=
\min(m,k)$.

\section{Three Definitions of Diversity Order and Their Relations}
\label{div_defs}

Consider the following definitions of diversity order: \vskip -7mm
\begin{subequations}
\begin{equation} \label{divdef1}
{\rm \overline{P}_e}(\rho)= ( G \rho )^{-d}+ o(\rho^{-d})
\end{equation}
\begin{equation} \label{divdef2}
\hspace{-12mm} \lim_{\rho \rightarrow \infty} \frac{{\rm
\overline{P}_e}(\tau \rho)} {{\rm \overline{P}_e}(\rho)}= \tau^{-d}
\end{equation}
\begin{equation} \label{divdef3}
\hspace{-9mm} \lim_{\rho \rightarrow \infty} \frac{\log {\rm
\overline{P}_e}(\rho)} {\log \rho}= -d
\end{equation} \vskip -3mm
\end{subequations}
\hspace{-5mm}where $G, \tau \in (0,\infty)$ are constants. As
mentioned in Section \ref{intro}, (\ref{divdef1}) and
(\ref{divdef3}) are adopted in existing literature (e.g. in
\cite{wang03} and \cite{zheng03} respectively), while
(\ref{divdef2}) is our preferred novel definition indicating that
${\rm \overline{P}_e}(\cdot)$ is regularly varying at $\infty$ with
exponent $-d$ when $d<\infty$, and rapidly varying with exponent
$-\infty$ when $d=\infty$. Note that the definition (\ref{divdef1})
requires $d \in (0,\infty)$, whereas (\ref{divdef2}) and
(\ref{divdef3}) allows for $d=\infty$ as well. We interpret the
rapid variation ($d=\infty$) in (\ref{divdef2}) or (\ref{divdef3})
as an average error rate which decays faster than $\rho^{-m}$, as
$\rho \rightarrow \infty$ for any $m \in (0,\infty)$. In other
words, $d=\infty$ can be interpreted as the ${\rm
\overline{P}_e}(\rho)$ versus $\rho$ plot on a log scale does not
become a straight line, but ``curves down'' as the average SNR
increases.

There has been some previous work on relating the channel PDF near
the origin to asymptotic error rates. The seminal work in
\cite{wang03} shows that if $f(z) \sim a z^m$ as $z \rightarrow 0$,
then (\ref{divdef1}) holds with diversity order $d= m+1$ and array
gain $G= \alpha ((2^{d-1} a \Gamma (d+1/2))/(\beta \sqrt{\pi}
d))^{-1/d}$. A relation between average error rate and outage
probability is also established by quantifying how the respective
diversity orders and array gains are related. The approach in
\cite{wang03} assumes {\bf (i)} the existence of the PDF and its
Maclaurin expansion near the origin, and is not naturally linked
with outage; {\bf (ii)} the instantaneous error rate is in the form
of a Q function (a special case of assumption {\bf AS1c}); {\bf
(iii)} the average error rate is of the form in (\ref{divdef1}) with
$d<\infty$.

As an example, it can be verified that $f(z) \sim 2 m^{2m} z^{m-1}
(-\log(m \sqrt{z})- \gamma_{\rm em} )/\Gamma^2(m)$ and $\lim_{z
\rightarrow 0} l(z)= \infty$ for generalized-$K$ fading with $m=k$,
and $f(z)$ does not admit a Maclaurin expansion, thus the approach
in \cite{wang03} is not applicable to this kind of channel. We
believe that the effective mathematical framework for studying
high-SNR asymptotic error rates is the Tauberian theorem, with a
distinct advantage of characterizing when outage events dominate
error rate performance by linking the CDF $F(\cdot)$ with ${\rm
\overline{P}_e}(\rho)$ at high average SNR. Our approach for
high-SNR analysis enables less restrictive assumptions on both the
instantaneous error rate and the channel distribution compared to
\cite{wang03}, as we will highlight in the sequel.

The following result establishes that, under general conditions,
diversity order defined through (\ref{divdef3}), is equivalent to
our definition (\ref{divdef2}) when the limit in (\ref{divdef2})
exists.

\begin{prop} \label{prop_def1}
If ${\rm \overline{P}_e}(\rho)$ is a Karamata function,
(\ref{divdef2}) is equivalent to (\ref{divdef3}) for $0<d \leq
\infty$. Moreover, definition (\ref{divdef1}) is a special case of
(\ref{divdef2}) and (\ref{divdef3}) (for $0<d<\infty$).
\end{prop}
\begin{proof}
See \ref{prop_def1_proof}.
\end{proof}

Since all average error rate expressions in existing literature are
Karamata functions, the novel definition in (\ref{divdef2}) is
equivalent to the general one in (\ref{divdef3}), and subsumes
(\ref{divdef1}).

\begin{spacing}{1.6}

\section{Asymptotic Analysis based on the Tauberian Theorem} \label{asympt}

In this section, we prove the asymptotic equivalence between the
high-SNR average error rate and the channel distribution at deep
fading (i.e. outage with small threshold), through the Tauberian
theorem. In Section \ref{asympt_gen}, we first establish a necessary
and sufficient condition for ${\rm \overline{P}_e}(\rho)=
\int_0^{\infty} {\rm P_e} (\rho z) dF(z)$ to exhibit diversity order
of $d$ for ${\rm P_e} (\rho z)$ satisfying {\bf AS1a}. Furthermore,
the relation between the asymptotic expressions of ${\rm
\overline{P}_e}(\rho)$ and $F(z)$ is also characterized through
regular variation for finite $d$. In Section \ref{ber_psk}, we
derive the high-SNR asymptotic average error rates for the
instantaneous error rate satisfying {\bf AS1b} or {\bf AS1c},
together with $F(z)$ satisfying {\bf AS2b} or {\bf AS2c}. These
results reveal a convenient and effective way to characterize
asymptotic error rate performance. Although the relation between
error rate and outage is addressed in several ways like in
\cite{wang03} and \cite{suraweera08}, we present brand new way to
relate average error rate to outage probability by expressing the
asymptotic average error rate in terms of a scaled version of outage
probability.

\subsection{Exponentially Bounded Instantaneous Error Rate} \label{asympt_gen}

For the general case with ${\rm P_e} (\rho z)$ satisfying {\bf
AS1a}, we have the following theorem which links the notion of
regular/rapid variation with diversity order, and allows for
$d=\infty$.
\begin{thry} \label{thm_div1}
For ${\rm P_e} (\rho z)$ satisfying {\bf AS1a}, ${\rm
\overline{P}_e}(\rho)$ exhibits a diversity order of $d \in
(0,\infty]$ if the CDF $F(z)$ of the channel power gain satisfies
{\bf AS2a}. The converse holds if $F(z)$ is a Karamata function (at
$0$).
\end{thry}
\begin{proof}
See \ref{thm_div1_proof}.
\end{proof}

Theorem \ref{thm_div1} fundamentally characterizes, with sufficient
and necessary conditions, the diversity order in terms of the CDF
(outage) for small arguments. This is unlike \cite{wang03} which
only provides sufficient conditions on the channel distribution to
achieve a certain diversity order. We now have the following theorem
which, unlike Theorem \ref{thm_div1}, rules out $d= \infty$ and
assumes $d< \infty$, but in return provides stronger results about
how the asymptotic expressions of ${\rm \overline{P}_e}(\rho)$ and
$F(z)$ are related, thereby characterizing when the outage event
dominates the error rate performance. \vskip 1mm
\begin{thry} \label{thm_asympt1}
For ${\rm P_e} (\rho z)$ satisfying {\bf AS1a}, if the CDF of the
channel power gain satisfies either {\bf AS2b} or {\bf AS2c}, ${\rm
\overline{P}_e}(\rho) \sim c_1 \rho^{-d} l(\rho^{-1})= c_1
F(\rho^{-1})$ as $\rho \rightarrow \infty$ where $c_1 \in
(0,\infty)$ is a constant. Conversely, if ${\rm
\overline{P}_e}(\rho)= \rho^{-d} r(\rho)$ as $\rho \rightarrow
\infty$ where $r(\rho)$ is slowly varying at $\infty$, $F(z) \sim
c_2 z^d r(z^{-1})= c_2 {\rm \overline{P}_e}(z^{-1})$ as $z
\rightarrow 0$ where $c_2 \in (0,\infty)$ is a constant, assuming
that $F(z)$ is a Karamata function (at $0$).
\end{thry} \vskip 1mm
\begin{proof}
See \ref{thm_asympt1_proof}.
\end{proof} \vskip 1mm

Theorem \ref{thm_asympt1} goes beyond characterizing diversity order
and points out the asymptotic proportionality of ${\rm
\overline{P}_e}(\rho)$ with $F(\rho^{-1})$ as $\rho \rightarrow
\infty$, for general ${\rm P_e} (\rho z)$ satisfying {\bf AS1a}.
This naturally establishes sufficient and necessary conditions on
the asymptotic equivalence between average error rate and outage. An
implication of Theorem \ref{thm_asympt1} is that for two different
communication systems over the same channel, there always exists a
constant SNR offset between their error rate performance at
sufficiently high average SNR, as long as their instantaneous error
rates are both exponentially bounded.

\subsection{Specific Instantaneous Error Rates} \label{ber_psk}

We have already established the regular variation of the CDF of the
channel power gain at $0$ as a necessary and sufficient condition
for a specific diversity order, under general modulation types with
${\rm P_e} (\rho z)$ satisfying {\bf AS1a}, using the Tauberian
theorem. We now offer the sharper results when ${\rm P_e} (\rho z)$
satisfies {\bf AS1b} or {\bf AS1c}, with the channel distribution
satisfying either {\bf AS2b} or {\bf AS2c}. The following results
make stronger assumptions about the instantaneous error rate ({\bf
AS1b}, {\bf AS1c}) but offer closed-form expressions for constants
$c_1$ and $c_2$ in Theorem \ref{thm_asympt1}.

For ${\rm P_e} (\rho z)$ satisfying {\bf AS1b}, the asymptotic
expression of the average error rate follows directly from Tauberian
theorem (Theorem \ref{tbr_th}) since the average error rate is in
the form of a Laplace-Stieltjes transform. Given {\bf AS2b} or {\bf
AS2c}, the asymptotic average error rate is given by \vskip -10mm
\goodbreak \begin{equation} \label{ber_dpsk_asymp} {\rm
\overline{P}_{e,AS1b}}(\rho):= \beta \int_0^{\infty} e^{-\alpha \rho
z} dF(z) \sim \beta \Gamma(d+1) (\alpha \rho)^{-d} l(\alpha^{-1}
\rho^{-1})= \beta \Gamma(d+1) F(\alpha^{-1} \rho^{-1})
\end{equation} \vskip -4mm
\hspace{-5mm}as $\rho \rightarrow \infty$, where the asymptotic
equality holds due to Tauberian theorem (Theorem \ref{tbr_th}) with
$F(z)$ corresponding to $H(x)$, and $\alpha \rho$ corresponding to
$s$. The last equality in (\ref{ber_dpsk_asymp}) holds since $F(z)=
z^d l(z)$. Equation (\ref{ber_dpsk_asymp}) shows the high-SNR
asymptotic average error rate in terms of the CDF $F(\cdot)$.
Conversely, given the average error rate ${\rm
\overline{P}_{e,AS1b}} (\rho)$ with diversity order $d$, the
asymptotic CDF of the channel power gain with small argument can be
obtained through the substitution $\rho= \alpha^{-1} z^{-1}$, based
on Tauberian theorem (Theorem \ref{tbr_th}).

For ${\rm P_e} (\rho z)$ satisfying {\bf AS1c}, the asymptotic
expression of the average error rate
\begin{equation} \label{ber_bpsk_int}
{\rm \overline{P}_{e,AS1c}}(\rho):= \int_0^{\infty}
\int_0^{\theta_{\rm t}} g_2(\theta) \exp \left( -\frac{\rho
z}{g_1(\theta)} \right) d \theta dF(z)= \int_0^{\theta_{\rm t}}
\int_0^{\infty} g_2(\theta) \exp \left( -\frac{\rho z}{g_1(\theta)}
\right) dF(z) d \theta
\end{equation}
can be obtained by evaluating $\lim_{\rho \rightarrow \infty} {\rm
\overline{P}_{e,AS1c}}(\rho)/ F(\rho^{-1})$. Note that in
(\ref{ber_bpsk_int}) we exchange the order of integrations due to
the positiveness of the integrand and the finiteness of the integral
\cite[p.457, C.9]{chaudhrybook01}. Like in (\ref{ber_dpsk_asymp}),
due to Tauberian theorem (Theorem \ref{tbr_th}) we have $G(\rho)=
\int_0^{\infty} \exp( -\rho z) dF(z) \sim \Gamma(d+1) F(\rho^{-1})$
as $\rho \rightarrow \infty$ and has variation exponent $-d$ at
$\infty$, then
\begin{equation}
\begin{split}
& \lim_{\rho \rightarrow \infty} \frac{ {\rm \overline{P}
_{e,AS1c}}(\rho)}{ F(\rho^{-1})}= \lim_{\rho \rightarrow \infty}
\frac{G(\rho)}{ F(\rho^{-1})} \cdot \lim_{\rho \rightarrow \infty}
\frac{ {\rm \overline{P} _{e,AS1c}}(\rho)}{G(\rho)}= \Gamma(d+1)
\lim_{\rho \rightarrow \infty} \frac{1}{G(\rho)} \int_0^{\theta_{\rm
t}} g_2(\theta) G\left( \frac{\rho}{g_1(\theta)} \right) d \theta=\\
& \Gamma(d+1) \int_0^{\theta_{\rm t}} \lim_{\rho \rightarrow \infty}
g_2(\theta) \frac{G(\rho/g_1(\theta))}{G(\rho)} d \theta=
\Gamma(d+1) \int_0^{\theta_{\rm t}} g_2(\theta) g_1^d(\theta) d
\theta
\end{split}
\end{equation} \vskip -4mm
\hspace{-5mm}where we changed the order of limit and integral in the
third equality based on the uniform convergence condition following
from \cite[Theorem 1.5.2]{binghambook89}, and consequently we obtain
\vskip -4mm
\goodbreak \begin{equation} \label{ber_bpsk_asymp}
{\rm \overline{P}_{e,AS1c}}(\rho) \sim \Gamma(d+1) \left(
\int_0^{\theta_{\rm t}} g_2(\theta) g_1^d(\theta) d \theta \right)
F(\rho^{-1})
\end{equation} \vskip -3mm
\hspace{-5mm}as $\rho \rightarrow \infty$. Like for
(\ref{ber_dpsk_asymp}), the converse of the result in
(\ref{ber_bpsk_asymp}) (from average error rate to asymptotic CDF)
is obtainable through the substitution $\rho= z^{-1}$. For many
practical cases, the closed-form expression of the integral
$\int_0^{\theta_{\rm t}} g_2(\theta) g_1^d(\theta) d \theta$ as a
function of $d$ can be worked out without an integral. For example,
by knowing the instantaneous error rates discussed in Section
\ref{chnl_sys_mdl}, we can obtain ${\rm \overline{P}_{e,BPSK}}(\rho)
\sim \Gamma(d+1/2) F(\rho^{-1})/(2 \sqrt{\pi})$ as $\rho \rightarrow
\infty$ for BPSK.

\end{spacing}

\vskip 1.5mm

In summary, for the two cases analyzed above, the asymptotic error
rate expressions in (\ref{ber_dpsk_asymp}) and
(\ref{ber_bpsk_asymp}) are given by $C_1 F(C_2 \rho^{-1}) \sim C_1
C_2^d F(\rho^{-1})$ with constants $C_1$ and $C_2$. Knowing that
outage event occurs when the instantaneous SNR $\rho z$ falls below
certain threshold, and since $F(C_2 \rho^{-1})$ is the probability
that $\rho z$ falls below $C_2$, the asymptotic error rate
expression $C_1 F(C_2 \rho^{-1})$ represents a scaled version of the
outage probability. Furthermore, $C_1$ and $C_2$ depend only on the
system specifications ($\alpha$, $\beta$, $\theta_{\rm t}$,
$g_1(\theta)$ and $g_2(\theta)$) and the variation exponent $d$ of
$F(z)$. As a result, we only need to know $d$ in addition to $F(z)$
to obtain the asymptotic error rate and do not need to express
$F(z)$ or the corresponding PDF $f(z)$ in series expansion form. For
simple practical channels like Nakagami-$m$, $d=m$ can be seen by
inspection of $F(z)$ or $f(z)$. For many channel distributions, the
diversity order can also be obtained by solving $\lim_{z \rightarrow
0} F(\tau z)/F(z)= \tau^d$ or $\lim_{z \rightarrow 0} f(\tau
z)/f(z)= \tau^{d-1}$ using L'H\^{o}pital's rule, or approximated
numerically by evaluating $\lim_{z \rightarrow 0} \log F(z)/\log z=
d$ or $\lim_{z \rightarrow 0} \log f(z)/\log z= d-1$. Therefore, the
asymptotic error rate is related to the channel distribution in many
useful ways which complement the PDF-based approach in
\cite{wang03}. In Figures \ref{bdpsk_m2_new} and \ref{bdpsk_m3_new},
we compare the BERs of DPSK and BPSK under Nakagami-$m$ fading
obtained through Monte Carlo simulation, with their approximations
given by (\ref{ber_dpsk_asymp}) and (\ref{ber_bpsk_asymp}). We
observe that the results given by (\ref{ber_dpsk_asymp}) and
(\ref{ber_bpsk_asymp}) match their corresponding simulation results
within $0.1$ dB at error rates of $10^{-6}$. In addition, it can be
seen that our approach gives better approximation to the Monte Carlo
simulation results than \cite{wang03}, most noticeably for moderate
values of average SNR.

\vskip 1.5mm

There are some practical cases of instantaneous error rate ${\rm
P_e} (\rho z)$, which do not fit {\bf AS1c} but can be expressed as
linear combinations of exponential-mixture functions of $\rho z$ in
the form $\int_0^{\theta_{\rm t}} g_2(\theta) \exp(-\rho
z/g_1(\theta)) d \theta$ given in {\bf AS1c}, such as the
instantaneous BER of Gray-coded $M$-PSK \cite[Section
8.1.1.3]{simonbook00}. There are also practical ${\rm P_e} (\rho z)$
which fits {\bf AS1c}, but can be expressed as a linear combination
of exponential-mixture functions, each with a much simpler
$g_2(\theta)$ function than that of ${\rm P_e} (\rho z)$ itself,
such as the instantaneous SER of square $M$-QAM. For these cases,
asymptotic characterization of ${\rm \overline{P}_{e}}(\rho)$
through linear combination becomes necessary. Specifically, for an
instantaneous error rate given by ${\rm P_e} (\rho z)= \sum_{j=1}^J
a_j {\rm P_e}_{,j} (\rho z)$ where $\{ a_j \}_{j=1}^J$ are constants
and $\{ {\rm P_e}_{,j} (\rho z) \}_{j=1}^J$ are (simpler)
expressions satisfying {\bf AS1b} or {\bf AS1c}, we first determine
${\rm E}[{\rm P_e}_{,j} (\rho z)] \sim C_{1j} F(C_{2j} \rho^{-1})$
using the methods in this section. If $\{ a_j \}_{j=1}^J$ are all
positive, we have ${\rm \overline{P}_{e}}(\rho) \sim \sum_{j=1}^J
a_j C_{1j} F(C_{2j} \rho^{-1})$; if any of $\{ a_j \}_{j=1}^J$ is
negative, ${\rm \overline{P}_{e}}(\rho) \sim \sum_{j=1}^J a_j C_{1j}
C_{2j}^d F(\rho^{-1})$ can be established. We omit the derivations
due to lack of space.

\section{Asymptotic Error Rate Performance under Diversity
Combining} \label{div_comb}

In this section, we establish an extension and application of the
results derived in Section \ref{ber_psk}, by analyzing the
asymptotic error rate performance at high average SNR for several
diversity combining schemes. The results are especially useful when
the diversity branches have non-identical fading distributions and
when the average error rate expression is not available in closed
form.

We consider a system in which the receiver has the channel state
information (CSI), and employs $N$ diversity branches with
independent but not necessarily identical fading distributions.
Particularly, the $n$-th branch is assumed to have a channel power
gain $z_n$ with CDF $F_n(z)=z^{d_n} l_n(z)$ where $l_n(z)$ is slowly
varying at $0$. We derive asymptotic expressions for E$[\beta
e^{-\alpha \rho z_{\rm c}}]$ and E$[\int_0^{\theta_{\rm t}}
g_2(\theta) \exp(-\rho z/g_1(\theta)) d \theta]$ for large $\rho$ in
terms of the system specifications ($\alpha$, $\beta$, $\theta_{\rm
t}$, $g_1(\theta)$ and $g_2(\theta)$), $\rho$, $\{ d_n \}_{n=1}^N$
and $\{ l_n(\cdot) \}_{n=1}^N$, where $z_{\rm c}$ is the channel
power gain after combining. In order to address this problem with a
unified approach, we first determine the asymptotic CDF of $z_{\rm
c}$ near the origin in the form $F_{\rm c}(z) \sim z^{d_{\rm c}}
l_{\rm c}(z)$ with $l_{\rm c}(z)$ being slowly varying at $0$ for
each specific diversity combining scheme by expressing $d_{\rm c}$
and $l_{\rm c}(\cdot)$ in terms of $\{ d_n \}_{n=1}^N$ and $\{
l_n(\cdot) \}_{n=1}^N$ respectively. Moreover, we express the CDF of
the combined channel $F_{\rm c}(z)$ in terms of $\{ d_n \}_{n=1}^N$
and $\{ F_n(\cdot) \}_{n=1}^N$. This will lead to characterizing
E$[\beta e^{-\alpha \rho z_{\rm c}}]$ and E$[\int_0^{\theta_{\rm t}}
g_2(\theta) \exp(-\rho z/g_1(\theta)) d \theta]$ in terms of the
system specifications, $\rho$, $d_{\rm c}$ and $l_{\rm c}(\cdot)$
(and alternatively $F_{\rm c}(\cdot)$) using the same method as in
Section \ref{ber_psk}.

\begin{spacing}{1.6}
\vspace{-3mm}

\subsection{Maximum Ratio Combining (MRC)} \label{cdf_mrc}

For MRC with independent branches as well as a number of cooperative
relay systems \cite{anghel04}, the channel power gain can be
expressed as the sum of independent random variables: $z_{\rm c}=
\sum_{n=1}^N z_n$. Define $\mathcal{L}_n(s):= \int_0^{\infty} e^{-s
z} dF_n(z)$ and $\mathcal{L}_{\rm c}(s):= \int_0^{\infty} e^{-s z}
dF_{\rm c}(z)$ to be the Laplace-Stieltjes transforms of the
respective distributions. Based on the convolution property of
Laplace transform, we have $\mathcal{L}_{\rm c}(s)= \prod_{n=1}^N
\mathcal{L}_n(s)$. On the other hand, it follows from Tauberian
theorem (Theorem \ref{tbr_th}) that $\mathcal{L}_n(s) \sim
\Gamma(d_n+1) s^{-d_n} l_n(s^{-1})$ as $s \rightarrow \infty$, and
therefore $\mathcal{L}_{\rm c}(s) \sim \prod_{n=1}^N \Gamma(d_n+1)
\cdot s^{-\sum_{n=1}^N d_n} \prod_{n=1}^N l_n(s^{-1})$. It is easy
to verify that $\prod_{n=1}^N \Gamma(d_n+1) \cdot \prod_{n=1}^N
l_n(s^{-1})$ is slowly varying at $\infty$ as a function of $s$. It
follows from Tauberian theorem (Theorem \ref{tbr_th}) that \vskip
-4mm
\goodbreak \begin{equation} \label{cdf_asympt_mrc}
F_{\rm c}(z) \sim \left[ \Gamma \left( \sum_{n=1}^N d_n+1 \right)
\right]^{-1} \prod_{n=1}^N \Gamma(d_n+1) \cdot \prod_{n=1}^N F_n(z)
\end{equation}
near the origin, where we have simplified $z^{\sum_{n=1}^N d_n}
\prod_{n=1}^N l_n(z)$ as $\prod_{n=1}^N F_n(z)$, and the first
factor on the right hand side is due to the variation exponent
$-\sum_{n=1}^N d_n$ of $\mathcal{L}_{\rm c}(s)$ at $\infty$. We will
relate the outage of the combined channel in (\ref{cdf_asympt_mrc})
to the average error rates in Section \ref{asympt_div_comb}.

\vspace{-3mm}

\subsection{Equal Gain Combining (EGC)}

For EGC with independent branches we have $z_{\rm c}= (\sum_{n=1}^N
\sqrt{z_n})^2/N$. Define $y_n=\sqrt{z_n}$ and $y_{\rm c}=
\sum_{n=1}^N y_n$. It is easy to derive that $y_n$ has a CDF given
by $G_n(y)= y^{2 d_n} l(y^2)$, and show that $l(y^2)$ is slowly
varying as a function of $y$ at $0$. Consequently, by using the same
method as in Section \ref{cdf_mrc}, we can derive the asymptotic CDF
$G_{\rm c}(y) \sim [ \Gamma ( 2\sum_{n=1}^N d_n+1 ) ]^{-1}
\prod_{n=1}^N \Gamma(2d_n+1) \cdot y^{2\sum_{n=1}^N d_n}
\prod_{n=1}^N l_n(y^2)$ of $y_{\rm c}$ near $y_{\rm c}=0$, then
using $z_{\rm c}= y_{\rm c}^2/N$, we get the asymptotic CDF \vskip
-4mm
\goodbreak \begin{equation} \label{cdf_asympt_egc}
F_{\rm c}(z) \sim \left[ \Gamma \left( 2\sum_{n=1}^N d_n+1 \right)
\right]^{-1} \prod_{n=1}^N \Gamma(2d_n+1) \cdot \prod_{n=1}^N
F_n(Nz)
\end{equation}
of $z_{\rm c}$ near the origin, which will be related to the
asymptotic error rates in Section \ref{asympt_div_comb}.

\vspace{-3mm}

\subsection{Selection Combining (SC)}

For SC with independent branches we have $z_{\rm c}=
\max_{n=1,2,...,N}\{ z_n \}$, then it follows that \vskip -4mm
\begin{equation} \label{cdf_asympt_sc}
F_{\rm c}(z)= \prod_{n=1}^N F_n(z)= z^{\sum_{n=1}^N d_n}
\prod_{n=1}^N l_n(z) ,
\end{equation} \vskip -2mm
\hspace{-5mm}which will be used to derive the asymptotic error rates
next.

\vspace{-3mm}

\subsection{Asymptotic Error Rate Expressions} \label{asympt_div_comb}

For all three cases analyzed above, we can verify that $F_{\rm
c}(z)$ is regularly varying at $0$ with exponent $d_{\rm
c}=\sum_{n=1}^N d_n>0$, since the same holds for $F_n(z)$ with
exponent $d_n>0$. In Section \ref{ber_psk} we established the
asymptotic average error rates in terms of the CDF $F(z)$ in
closed-form. We can directly apply these results here to the
combined CDFs of the respective diversity combining schemes, and
obtain \vskip -4mm
\begin{equation} \label{ser_asympt_expf}
{\rm E}[\beta e^{-\alpha \rho z_{\rm c}}] \sim \beta \Gamma(d_{\rm
c}+1) F_{\rm c}(\alpha^{-1} \rho^{-1})
\end{equation}
\begin{equation} \label{ser_asympt_qf}
{\rm E} \left[ \int_0^{\theta_{\rm t}} g_2(\theta) \exp \left(
-\frac{\rho z}{g_1(\theta)} \right) d \theta \right] \sim \left(
\Gamma(d_{\rm c}+1) \int_0^{\theta_{\rm t}} g_2(\theta) g_1^{d_{\rm
c}}(\theta) d \theta \right) F_{\rm c}(\rho^{-1})
\end{equation}
as $\rho \rightarrow \infty$, where $d_{\rm c}$ is the same for all
three combining schemes, and $F_{\rm c}(\cdot)$ is given by
(\ref{cdf_asympt_mrc}), (\ref{cdf_asympt_egc}) or
(\ref{cdf_asympt_sc}) for MRC, EGC and SC respectively. We have thus
established a unified approach to evaluate the asymptotic error rate
at high average SNR for MRC, EGC and SC, given the conditions that
the channel power gain of each branch has a CDF which is regularly
varying at $0$ and the instantaneous error rate can be expressed as
linear combination of exponential or exponential-mixture functions.
Figures \ref{bpsk_div_comb1} and \ref{bpsk_div_comb2} show the BERs
of BPSK under different diversity combining schemes obtained through
Monte Carlo simulation, as well as their approximations obtained
through (\ref{ser_asympt_qf}) together with (\ref{cdf_asympt_mrc}),
(\ref{cdf_asympt_egc}) and (\ref{cdf_asympt_sc}). The accuracy of
our approach is corroborated by the closeness (within $0.2$ dB at
error rate of $10^{-6}$) between the simulation results and
analytical asymptotic approximations. Like observed in Figures
\ref{bdpsk_m2_new} and \ref{bdpsk_m3_new}, our approach gives better
approximation than \cite{wang03} to the simulation results,
especially when the average SNR is not significantly high.

\end{spacing}

In addition to the simple relation between the asymptotic error rate
and the channel distribution, it can be seen that our approach based
on Tauberian theorem enables analysis of the performance of a
communication system involving multiple channels, in which the
regular variation property of the distribution of the overall
effective channel is inherited from the channel distributions
corresponding to the constituent parts of the system. Also, the
asymptotic CDFs given by (\ref{cdf_asympt_mrc}) and
(\ref{cdf_asympt_egc}) can make high-SNR substitutes of the
numerical inversion method in \cite{ko00} to compute outages in
fading channels, as long as different diversity branches are
independent.

\section{Conclusions} \label{concl}

In this paper, we investigate the relationship between the high-SNR
asymptotic average error rate and outage by establishing that the
instantaneous error rate being upper bounded by an exponential
function of the instantaneous SNR is sufficient for outage events to
dominate error rate performance. We do this by establishing the
regular/rapid variation of the CDF of the channel power gain with
exponent $d$ at $0$ to be a necessary and sufficient condition for
the error rate to exhibit a diversity order of $d$. This implies the
equivalence between the error-rate-based diversity order and the
outage-based diversity order, thereby characterizing the conditions
under which the outage event dominates high-SNR error rate
performance. For the case of finite $d$, the constant SNR offset
between different communication systems over the same fading channel
with exponentially bounded instantaneous error rates is revealed.
For instantaneous error rates given by a mixture of exponential
functions of the instantaneous SNR ({\bf AS1c}), we express the
high-SNR asymptotic error rates directly as multiples of the outage
probability (channel CDF for small arguments). Furthermore, we
derive the asymptotic error rate for instantaneous error rate
satisfying {\bf AS1c}, considering different diversity combining
schemes with the CDF of the channel power gain of each diversity
branch being regularly varying at $0$. All these results exhibit the
convenience and effectiveness of Tauberian theorem as a tool to
analyze the asymptotic error rate performance related to diversity
under fading, thus conveniently generalizing and complementing the
PDF-based approach in \cite{wang03}. Numerical results show that our
approach gives more accurate approximations than \cite{wang03}.

\renewcommand\thesection{Appendix \Alph{section}}
\setcounter{section}{0}

\vspace{-3mm}

\begin{spacing}{1.6}

\section{Proof of Proposition \ref{prop_sl1}} \label{prop_sl1_proof}

We prove this for $H(x)$ being slowly varying at $\infty$ only,
since the case of slow variation at $0$ can be obtained by
considering $H(x^{-1})$. By definition we have $\lim_{x \rightarrow
\infty} H(tx)/H(x)= 1$ for $t>0$, which implies that for any
$\epsilon_1>0$, there exists sufficiently large $s_1$ such that for
all $x>s_1$, $|H(tx)/H(x)-1|<\epsilon_1$. Without loss of generality
we assume $t>1$, and since $0 \leq h_1 \leq H(x) \leq h_2 < \infty$
for sufficiently large $x$ (say $x>s_2$), it follows that for any
$\epsilon_2= \epsilon_1 h_2$, there exists sufficiently large
$s_3=\max \{ s_1,s_2 \}$ such that for all $x_1>s_3$ and $x_2>s_3$,
$|H(x_2)-H(x_1)|< \epsilon_2$. This follows from multiplying
$|H(tx)/H(x)-1|<\epsilon_1$ by $H(x)$. Therefore the Cauchy criteria
for the existence of limit is satisfied, and $H(x)$ should converge
to a constant in $[h_1,h_2]$ as $x \rightarrow \infty$.

\vspace{-2mm}

\section{Proof of Proposition \ref{prop_def1}} \label{prop_def1_proof}

We use Theorem \ref{rep_thm} to prove the relation between
(\ref{divdef2}) and (\ref{divdef3}), starting from the claim that
(\ref{divdef2}) implies (\ref{divdef3}). For $0<d \leq \infty$,
(\ref{divdef2}) implies that ${\rm \overline{P}_e}(\rho)$ has a
variation exponent $-d$ at $\infty$, and thus can be represented as
\vskip -5mm
\goodbreak \begin{equation} \label{var_rep1}
{\rm \overline{P}_e}(\rho)= c(\rho) \exp \left\{ \int_a^{\rho}
\frac{\epsilon(u)}{u} du \right\}
\end{equation} \vskip -2mm
\hspace{-5mm}for some $a>0$, $c(\rho) \rightarrow c \in (0,\infty)$
and $\epsilon(\rho) \rightarrow -d$ as $\rho \rightarrow \infty$,
using Theorem \ref{rep_thm}. Therefore \vskip -6mm
\begin{equation} \label{lim1}
\begin{split}
\lim_{\rho \rightarrow \infty} \frac{\log {\rm
\overline{P}_e}(\rho)}{\log \rho} &= \lim_{\rho \rightarrow \infty}
\frac{\log c(\rho)+ \int_a^{\rho} \frac{\epsilon(u)}{u} du}{\log
\rho}= \lim_{\rho \rightarrow \infty} \frac{\int_a^{\rho}
\frac{\epsilon(u)}{u}
du}{\log \rho}\\
&= \lim_{\rho \rightarrow \infty} \frac{\epsilon(\rho)}{\rho}
(\rho^{-1})^{-1}= \lim_{\rho \rightarrow \infty} \epsilon(\rho)=-d
\end{split}
\end{equation} \vskip -2mm
\hspace{-5mm}where in the second equality we used $\epsilon(\rho)
\rightarrow -d$, and in the third equality L'H\^{o}pital's rule. We
have thus shown that (\ref{divdef2}) implies (\ref{divdef3}).

We next prove that (\ref{divdef3}) implies (\ref{divdef2}) given the
mild additional assumption that the limit in (\ref{divdef2}) exists.
The existence of the limit in (\ref{divdef2}) implies that
$\lim_{\rho \rightarrow \infty} {\rm \overline{P}_e}(\tau \rho)/
{\rm \overline{P}_e}(\rho)= \tau^{m_1}$ for some $m_1 \in [-\infty,
\infty]$ \cite[Lemma 1, p.275]{fellerbook71}. We would like to show
that $m_1=-d$. Clearly, ${\rm \overline{P}_e}(\rho)$ can be
represented by (\ref{var_rep1}) for some $a>0$, $c(\rho) \rightarrow
c \in (0,\infty)$ and $\epsilon(\rho) \rightarrow m_1$ as $\rho
\rightarrow \infty$. Similar to (\ref{lim1}), it can be derived that
$\lim_{\rho \rightarrow \infty} (\log {\rm
\overline{P}_e}(\rho))/(\log \rho)=m_1$. Given the condition that
(\ref{divdef3}) holds, it follows that $m_1=-d$, and thus
(\ref{divdef2}) also holds.

To show that (\ref{divdef1}) is a special case of (\ref{divdef2}),
from (\ref{divdef1}) we get ${\rm \overline{P}_e}(\tau \rho)= G^{-d}
\tau^{-d} \rho^{-d}+ o(\rho^{-d})$, then
\begin{equation}
\lim_{\rho \rightarrow \infty} \frac{{\rm \overline{P}_e}(\tau
\rho)} {{\rm \overline{P}_e}(\rho)}= \lim_{\rho \rightarrow \infty}
\frac{G^{-d} \tau^{-d}+ o(\rho^{-d})/ \rho^{-d}}{G^{-d}+
o(\rho^{-d})/ \rho^{-d}}= \tau^{-d}
\end{equation} \vskip -2mm
\hspace{-5mm}i.e. (\ref{divdef2}) holds. An example for why
(\ref{divdef2}) does not imply (\ref{divdef1}) can be seen for the
case ${\rm \overline{P}_e}(\rho) \sim c_1 (\log \rho)^{c_2}
\rho^{-d}$ with $c_1,c_2>0$, as brought up in
\cite{taherzadeh07a,taherzadeh10}. Hence (\ref{divdef2}) is more
general than (\ref{divdef1}).

\vspace{-3mm}

\section{Proof of Theorem \ref{thm_div1}} \label{thm_div1_proof}

We begin with assuming the regular/rapid variation of $F(z)$ with
exponent $d$ at $0$, and proving the regular/rapid variation of
${\rm \overline{P}_e}(\rho)$ with exponent $-d$ at $\infty$, by
analyzing the upper and lower bounds of ${\rm
\overline{P}_e}(\rho)$. Due to {\bf AS1a} we have \vskip -5mm
\goodbreak \begin{equation} \label{ber_ub}
{\rm \overline{P}_e}(\rho)= \int_0^{\infty} {\rm P_e} (\rho z) dF(z)
\leq \int_0^{\infty} \beta e^{-\alpha \rho z} dF(z)= {\rm
\overline{P}_{e,UB}}(\rho)
\end{equation} \vskip -3mm
\hspace{-5mm}and since ${\rm P_e} (\cdot)$ is non-negative and
monotonically decreasing, \vskip -5mm
\goodbreak \begin{equation} \label{ber_lb}
{\rm \overline{P}_e}(\rho) \geq \int_0^{\eta/\rho} {\rm P_e} (\rho
z) dF(z) \geq \int_0^{\eta/\rho} {\rm P_e} (\eta) dF(z)= {\rm P_e}
(\eta) F(\eta/\rho)= {\rm \overline{P}_{e,LB}}(\rho)
\end{equation} \vskip -3mm
\hspace{-5mm}for any constant $\eta>0$. It follows from
(\ref{ber_ub}) and (\ref{ber_lb}) that \vskip -3mm
\begin{equation} \label{div_bound}
\frac{\log {\rm \overline{P}_{e,LB}}(\rho)} {\log \rho} \leq
\frac{\log {\rm \overline{P}_e}(\rho)} {\log \rho} \leq \frac{\log
{\rm \overline{P}_{e,UB}}(\rho)} {\log \rho}
\end{equation} \vskip -2mm

For $d=\infty$, it can be easily seen that ${\rm
\overline{P}_{e,LB}}(\rho)$ has variation exponent $-\infty$ at
$\rho= \infty$ and thus satisfies (\ref{divdef3}). Also, based on
Tauberian theorem (Theorem \ref{tbr_th}) with the correspondences
$F(z) \equiv H(x)$ and $\alpha \rho \equiv s$, ${\rm
\overline{P}_{e,UB}}(\rho)$ has variation exponent $-\infty$ at
$\rho= \infty$ and thus satisfies (\ref{divdef3}). By taking the
limit of (\ref{div_bound}) as $\rho \rightarrow \infty$ and using
Squeeze theorem (a well-known theorem stating that if $\lim_{x
\rightarrow a} H_1(x)= \lim_{x \rightarrow a} H_3(x)= h$ and $H_1(x)
\leq H_2(x) \leq H_3(x)$ in some neighborhood of $a$, then $\lim_{x
\rightarrow a} H_2(x)= h$), it follows that ${\rm
\overline{P}_e}(\rho)$ should satisfy (\ref{divdef3}).

For $0<d< \infty$, let $F(z)= z^d l(z)$ with $l(z)$ being slowly
varying at $0$. Based on (\ref{ber_ub}) and Tauberian theorem
(Theorem \ref{tbr_th}) we have \vskip -5mm
\begin{equation} \label{ber_ub2}
{\rm \overline{P}_{e,UB}}(\rho) \sim \beta \Gamma(d+1) (\alpha
\rho)^{-d} l(\alpha^{-1} \rho^{-1}) .
\end{equation} \vskip -3mm
\hspace{-5mm}Also, it follows from (\ref{ber_lb}) that \vskip -5mm
\begin{equation} \label{ber_lb2}
{\rm \overline{P}_{e,LB}}(\rho)= \eta^d \rho^{-d} {\rm P_e} (\eta)
l(\eta/\rho) .
\end{equation} \vskip -3mm
\hspace{-5mm}It is easy to verify that both ${\rm
\overline{P}_{e,UB}}(\rho)$ and ${\rm \overline{P}_{e,LB}}(\rho)$
satisfy (\ref{divdef2}), thus also satisfy (\ref{divdef3}) due to
Proposition \ref{prop_def1}, then ${\rm \overline{P}_e}(\rho)$
should satisfy (\ref{divdef3}) due to (\ref{div_bound}). This
completes the sufficiency part.

We now show ${\rm \overline{P}_e}(\rho)$ being regularly/rapidly
varying with exponent $-d$ at $\infty$ implies $F(z)$ being
regularly/rapidly varying with exponent $d$ at $0$ provided that
$F(z)$ is a Karamata function (at $0$). Based on \cite[Lemma 1,
p.275]{fellerbook71}, it can be confirmed that $\lim_{z \rightarrow
0} F(\tau z)/F(z)= \tau^{m_2}$ with $m_2 \in (0,\infty]$ being the
variation exponent of $F(z)$. Like in the sufficiency part of the
proof, it can be derived that both ${\rm \overline{P}_{e,UB}}(\rho)$
and ${\rm \overline{P}_{e,LB}}(\rho)$ have variation exponent
$-m_2$, then (\ref{div_bound}) becomes $-m_2 \leq -d \leq -m_2$,
leading to $m_2=d$. Therefore $F(z)$ is regularly/rapidly varying
with exponent $d$ at $0$.

\end{spacing}

\vspace{-3mm}

\section{Proof of Theorem \ref{thm_asympt1}} \label{thm_asympt1_proof}

We first derive the asymptotic expression for ${\rm
\overline{P}_e}(\rho)$ assuming {\bf AS2b}, $F(z)=z^d l(z)$. Theorem
\ref{thm_div1} implies that ${\rm \overline{P}_e}(\rho)= \rho^{-d}
r(\rho)$ where $r(\rho)$ is slowly varying at $\infty$. Let
$l_1(\rho):= {\rm \overline{P}_e}(\rho)/(\rho^{-d} l(\rho^{-1}))$.
It is easy to show that $l_1(\rho)$ is slowly varying at $\infty$.
On the other hand, from (\ref{ber_ub}) and (\ref{ber_lb}) we have
${\rm \overline{P}_{e,LB}}(\rho)/(\rho^{-d} l(\rho^{-1})) \leq
l_1(\rho) \leq {\rm \overline{P}_{e,UB}}(\rho)/(\rho^{-d}
l(\rho^{-1}))$; also, based on (\ref{ber_ub2}) and (\ref{ber_lb2})
as well as the definition of slow variation, we have $\lim_{\rho
\rightarrow \infty} {\rm \overline{P}_{e,UB}}(\rho)/(\rho^{-d}
l(\rho^{-1}))= \beta \alpha^{-d} \Gamma(d+1)$ and $\lim_{\rho
\rightarrow \infty} {\rm \overline{P}_{e,LB}}(\rho)/(\rho^{-d}
l(\rho^{-1}))= {\rm P_e}(\eta) \eta^d$ respectively. It follows that
for sufficiently large $\rho$, $l_1(\rho)$ is bounded between two
finite positive constants. Based on Proposition \ref{prop_sl1},
$l_1(\rho)$ converges to a finite positive constant as $\rho
\rightarrow \infty$, say $l_1(\rho) \sim c_1 \in (0,\infty)$, and
hence ${\rm \overline{P}_e}(\rho) \sim c_1 \rho^{-d} l(\rho^{-1})$.

Consider now the derivation for the asymptotic expression of $F(z)$
assuming ${\rm \overline{P}_e}(\rho)= \rho^{-d} r(\rho)$ and that
$\lim_{z \rightarrow 0} F(\tau z)/F(z)$ exists for $\tau>0$. It
follows directly from Theorem \ref{thm_div1} that $F(z)$ must be in
the form $F(z)= z^d l(z)$ with $l(z)$ slowly varying at $0$. Based
on the sufficiency part of the proof (i.e. the asymptotic expression
of ${\rm \overline{P}_e}(\rho)$ given $F(z)=z^d l(z)$), we must have
${\rm \overline{P}_e}(\rho)/(\rho^{-d} l(\rho^{-1}))=
r(\rho)/l(\rho^{-1}) \sim c_1$ as $\rho \rightarrow \infty$, with
constant $c_1 \in (0,\infty)$. Consequently, $l(z)/r(z^{-1}) \sim
c_1^{-1}$ as $z \rightarrow 0$, and $F(z) \sim c_2 z^d r(z^{-1})$
where $c_2= c_1^{-1} \in (0,\infty)$ is a constant.

\vspace{-2mm}

\begin{spacing}{1.4}

\bibliographystyle{IEEEtran}
\bibliography{reflist}

\end{spacing}

\begin{figure}[!ht]
\begin{minipage}{1.0\textwidth}
\begin{center}
\includegraphics[height=9cm,keepaspectratio]{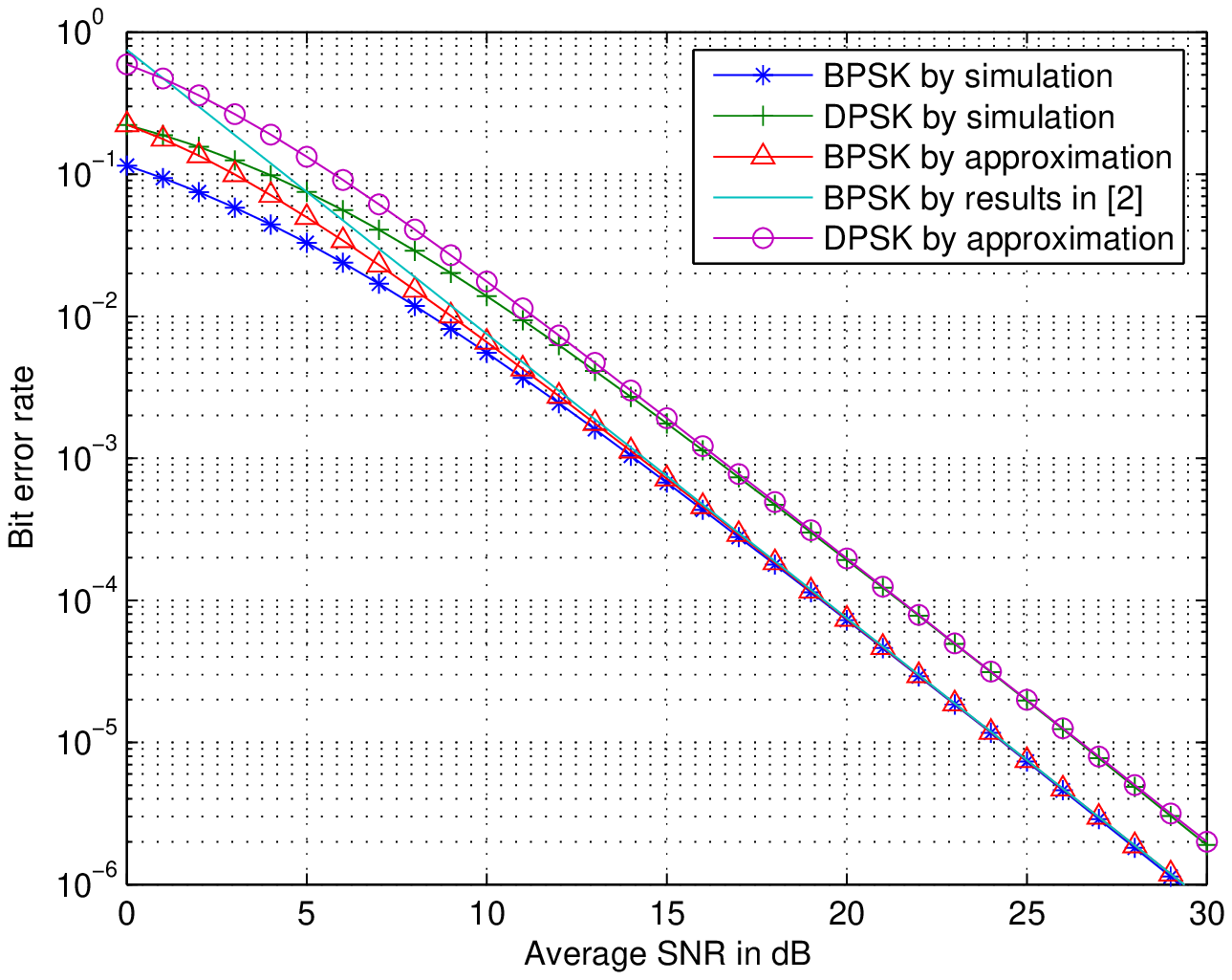}
\caption{BERs of BPSK and DPSK under Nakagami-$m$ fading with $m=2$}
\label{bdpsk_m2_new}
\end{center}
\end{minipage}
\end{figure}

\begin{figure}[!ht]
\begin{minipage}{1.0\textwidth}
\begin{center}
\includegraphics[height=9cm,keepaspectratio]{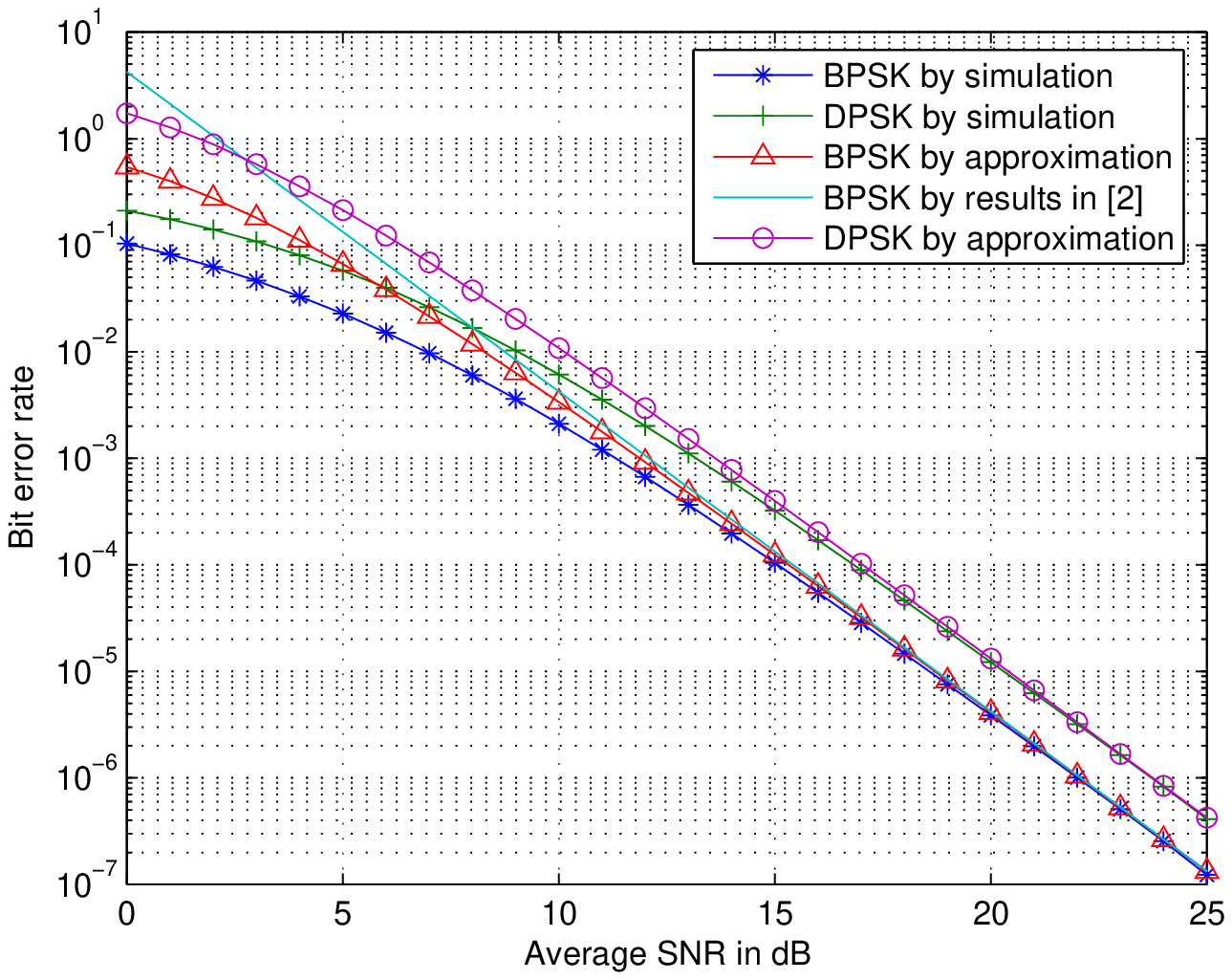}
\caption{BERs of BPSK and DPSK under Nakagami-$m$ fading with $m=3$}
\label{bdpsk_m3_new}
\end{center}
\end{minipage}
\end{figure}

\begin{figure}[!ht]
\begin{minipage}{1.0\textwidth}
\begin{center}
\includegraphics[height=9cm,keepaspectratio]{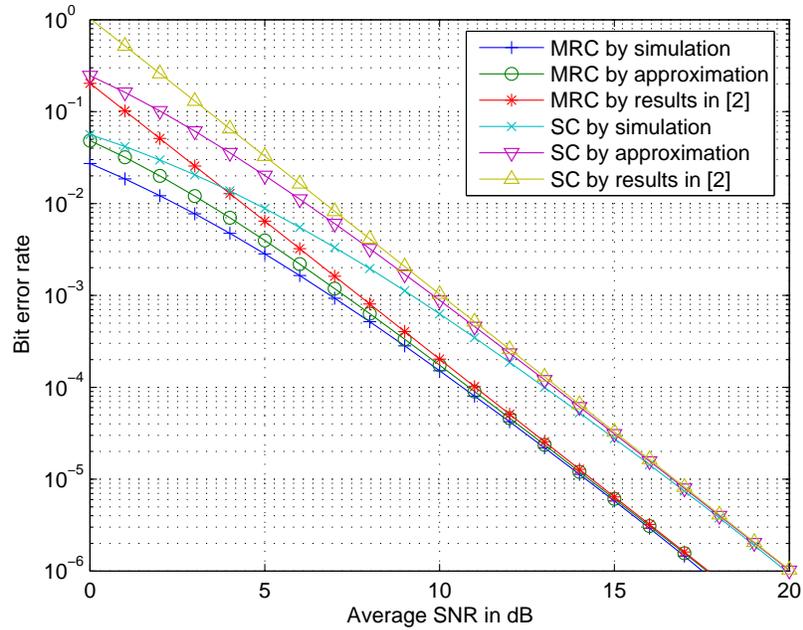}
\caption{BERs of BPSK under MRC and SC involving three Nakagami-$m$
fading branches with $m=0.5$, $1$ and $1.5$} \label{bpsk_div_comb1}
\end{center}
\end{minipage}
\end{figure}

\begin{figure}[!ht]
\begin{minipage}{1.0\textwidth}
\begin{center}
\includegraphics[height=9cm,keepaspectratio]{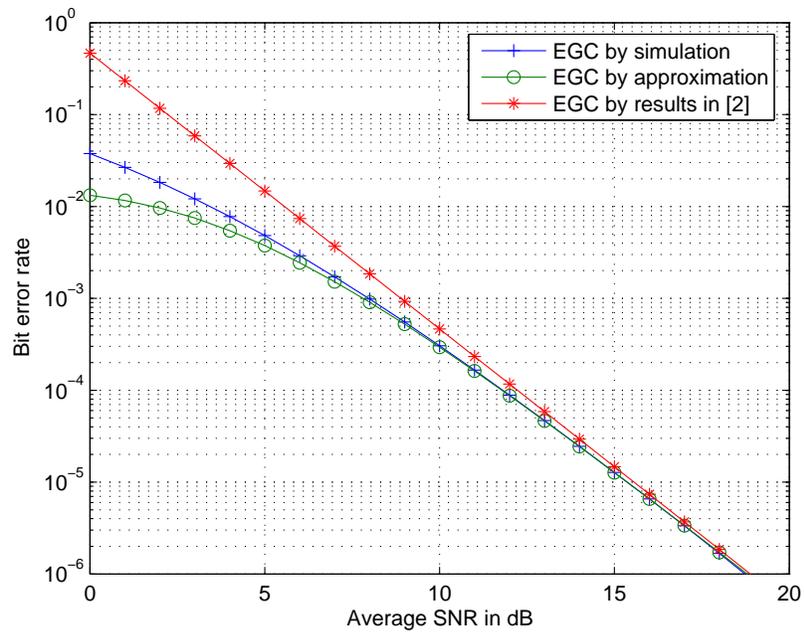}
\caption{BERs of BPSK under EGC involving three Nakagami-$m$ fading
branches with $m=0.5$, $1$ and $1.5$} \label{bpsk_div_comb2}
\end{center}
\end{minipage}
\end{figure}

\end{document}